\documentclass[a4paper,UKenglish]{lipics-v2021}
\usepackage[utf8]{inputenc}
\usepackage{microtype}
\usepackage{amsmath,amssymb,amsthm}
\usepackage{xspace}

\usepackage[ruled, vlined, linesnumbered]{algorithm2e}
\SetArgSty{textnormal}

\usepackage{amsfonts}
\usepackage{mathtools}

\hideLIPIcs 
\nolinenumbers

\usepackage[nocompress]{cite}

\title{Resilient Level Ancestor, Bottleneck, and Lowest Common Ancestor Queries in Dynamic Trees}
\titlerunning{Resilient Level Ancestors, Bottleneck, and LCA Queries in Dynamic Trees}

\newcommand{\dept}{Dept.\kern 0.2883em}

\author{Luciano~Gualà}{Department of Enterprise Engineering, University of Rome ``Tor Vergata'', Italy}{guala@mat.uniroma2.it}{https://orcid.org/0000-0001-6976-5579}{This work was partially supported by the project E89C20000620005 ``ALgorithmic aspects of BLOckchain TECHnology'' (ALBLOTECH).}

\author{Stefano Leucci}{\dept of Information Engineering, Computer Science, and Mathematics, University of L'Aquila, Italy\and \url{https://www.stefanoleucci.com}}{stefano.leucci@univaq.it}{https://orcid.org/0000-0002-8848-7006}{}

\author{Isabella Ziccardi}{\dept of Information Engineering, Computer Science, and Mathematics, University of L'Aquila, Italy}{isabella.ziccardi@graduate.univaq.it}{https://orcid.org/0000-0002-1550-3677}{}

\authorrunning{L. Gualà, S. Leucci, and I. Ziccardi} 

\Copyright{Luciano Gualà, Stefano Leucci, and Isabella Ziccardi}



\usepackage{xspace}
\newcommand{\LA}{\ensuremath{\texttt{LA}}\xspace}
\newcommand{\RMQ}{\ensuremath{\texttt{RMQ}}\xspace}
\newcommand{\LCA}{\ensuremath{\texttt{LCA}}\xspace}
\newcommand{\BVQ}{\ensuremath{\texttt{BVQ}}\xspace}
 
\newcommand{\AddLeaf}{\ensuremath{\texttt{AddLeaf}}\xspace}
\newcommand{\NewTree}{\ensuremath{\texttt{NewTree}}\xspace}

\newcommand{\parent}{\ensuremath{\mathrm{parent}}\xspace}
\newcommand{\depth}{\ensuremath{\mathrm{depth}}\xspace}

\newcommand{\mynull}{\ensuremath{\mathtt{null}}\xspace}

\newcommand{\climb}{\texttt{climb}}

\newcommand{\nodeflag}{\ensuremath{\mathrm{flag}}\xspace}

\newcommand{\nearblack}{\ensuremath{\mathrm{near\_black}}\xspace}

\newcommand{\A}{\mathcal{A}}
\newcommand{\pseudoT}{\mathcal{T}}

\begin{document}
\maketitle

\begin{abstract}
We study the problem of designing a \emph{resilient} data structure maintaining a tree under the Faulty-RAM model [Finocchi and Italiano, STOC'04] in which up to $\delta$ memory words can be corrupted by an adversary. 
Our data structure stores a rooted dynamic tree that can be updated via the addition of new leaves, requires linear size, and supports \emph{resilient} (weighted) level ancestor queries, lowest common ancestor queries, and bottleneck vertex queries in $O(\delta)$ worst-case time per operation.
\end{abstract}

\section{Introduction}

Due to a diverse spectrum of reasons, ranging from manufacturing defects to charge collection, the data stored in modern memories can sometimes face corruptions, a problem that is exacerbated by the recent growth in the amount of stored data. To make matters worse, even a single memory corruption can cause classical algorithms and data structures to fail catastrophically.
One mitigation approach relies on low-level error-correcting schemes that transparently detect and correct such errors.
These schemes however either require expensive hardware or employ space-consuming replication strategies.
Another approach, which has recently received considerable attention \cite{FI04,FGI07,BFFGIJMM07,JMM07,Italiano10,PetrilloGI13,LLM19}, aims to design \emph{resilient} algorithms and data structures that are able to remain operational even in the presence of memory faults, at least with respect to the set of uncorrupted values.

In this paper we consider the problem of designing resilient data structures that store a \emph{dynamic} rooted tree $T$ while answering several types of queries. 
More formally, we focus on maintaining a tree that initially consists of a single vertex (the root of the tree) and can be dynamically augmented via the $\AddLeaf(v)$ operation that appends a new leaf as a child of an existing vertex $v$.\footnote{In the literature this setting is also called \emph{incremental} or \emph{semi-dynamic} to emphasize that arbitrary insertions and deletions of tree vertices/edges are not supported.
In this paper, unless otherwise specified, we follow the terminology of \cite{D91} by considering \emph{dynamic} trees that only support insertion of leaves.}
It is possible to \emph{query} $T$ in order to obtain information about its current topology. 
We mainly concerned on the following well-known query types:
\begin{description}
\item[(Weighted) Level Ancestor Queries:] Given a vertex $v$ and an integer $k$, the query $\LA(v,k)$ returns the \emph{$k$-parent of $v$}, i.e., the vertex at distance $k$ from $v$ among the ancestors of $v$.
In the weighted version of the problem each vertex of the tree $T$ is associated with a small (polylogarithmic) positive integer weight, and a query needs to report the closest ancestor $u$ of $v$ such that the total weight of the path from $v$ to $u$ in $T$ is at least $k$.
\item[Lowest Common Ancestor Queries:] Given two vertices $u$, $v$, the query $\LCA(u,v)$ returns the vertex at maximum depth in $T$ that is simultaneously and ancestor of both $u$ and $v$.
\item[Bottleneck Queries:] In this problem, each vertex has an associated  integer weight and, given two vertices $u$, $v$, a $\BVQ(u,v)$ query reports the minimum/maximum-weight vertex in the path between $u$ and $v$ in $T$.\footnote{It is easy to see that this also captures the well-known \emph{bottleneck edge query} variant \cite{DemaineLW14}, in which weights are placed on edges instead of vertices.}
It is worth noticing that, when $T$ is a path, the above problem can be seen as a dynamic version of the classical \emph{range minimum query} problem.
In the range minimum query problem, a query $\RMQ(i,j)$ reports the minimum element between the $i$-th and the $j$-th element of a (static) input sequence \cite{BenderF00}.
\end{description}

For all of the above problems, \emph{linear-size} \emph{non-resilient} data structures supporting both the $\AddLeaf$ and the query operations in constant worst-case time are known \cite{AH00,CH05}.
It is then natural to investigate what can be achieved for the above problem when the sought data structures are required to withstand memory faults.

To precisely capture the behaviour of resilient algorithms, one needs to employ a model of computation that takes into account potential memory corruptions.
To this aim, we adopt the \emph{Faulty-RAM} model introduced by Finocchi and Italiano in \cite{FI04}. This model is similar to the classical RAM model except that all but $O(1)$ memory words can be subject to corruptions that alter theirs contents to an arbitrary value. The overall number of corruptions is upper bounded by a parameter $\delta$ and such corruptions are chosen in a \emph{worst-case} fashion by a computationally unbounded \emph{adversary}.

A simple error-correcting strategy based on replication provides a general scheme for obtaining resilient versions of any classical \emph{non-resilient} data structure at a cost of a $\Theta(\delta)$ blowup in both the time needed for each operation and the size of the data structure.
This space overhead is undesirable, especially when $\delta$ can be large.
For the above reason, the main goal in the area is obtaining compact solutions with a particular focus on linear-size data structures \cite{FGI09,FI04,BFFGIJMM07,FGI07,Italiano10,PetrilloGI13,FinocchiGI09,JMM07}. 
However, for linear-size data structures, even $\delta = \omega(1)$ corruptions can be already sufficient for the adversary to irreversibly corrupt some of the stored elements \cite{FGI09}.
The solution adopted in the literature is that of suitable relaxing the notion of correctness by only requiring queries to answer correctly with respect to the portions of the data structure that are uncorrupted.
Notice that this is not easy to obtain since corruptions in unrelated parts of the data structure can still misguide the execution of a query (see \cite{FGI09} for a discussion).\footnote{For example, the authors of \cite{FGI09} consider the problem of designing linear-size resilient dictionaries adopt a notion of \emph{resilient search} that requires the search procedure to answer correctly w.r.t. all uncorrupted keys (see Section~\ref{sec:related} for a more precise definition). Notice how the classical solutions based on search trees do not meet this requirement since a single unrelated corruption can destroy the tree path leading to the sought key.}

\subsection{Our results}

\begin{figure}
    \centering
    \includegraphics[width=0.8\textwidth]{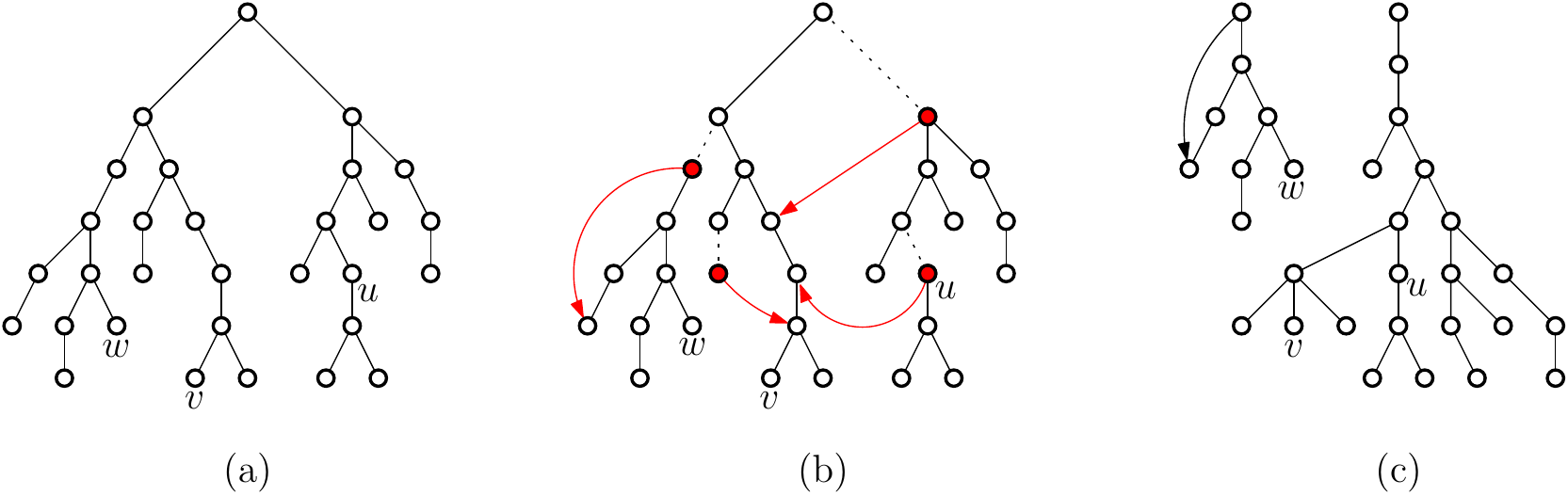}
    \caption{Illustration of resilient LA queries. The current tree $T$ logically maintained by the data structure is depicted in (a). In this example, each vertex maintains a reference to its parent in $T$.
    In (b) some of the parent-child relationships have been altered by the adversary by corrupting the nodes highlighted in red.
    Since the algorithm cannot distinguish corrupted memory words from uncorrupted ones, its (defective) view of $T$ is shown in (c).
    Nevertheless, a resilient data structure must still be able to correctly answer queries involving uncorrupted paths.
    For exampple, the query $\LA(v,k)$ is required to answer correctly for all (meaningful) values of $k$ since the path from $v$ to the root is uncorrupted, while query $\LA(w,k)$ is required to answer correctly for $k \leq 2$.
    Since $u$ is corrupted, the query $\LA(u,k)$ is allowed to answer incorrectly for every value of $k$.}
    \label{fig:resilientLA}
\end{figure}

We design a data structure maintaining a dynamic tree that can be updated via the addition of new leaves, and supports \emph{resilient} (weighted) $\LA$, $\LCA$, and $\BVQ$ queries.

Our data structure stores each vertex of the current tree $T$ in a single memory word of $\Theta(\log n)$ bits. We will say that a vertex $v$ is corrupted if the memory word associated with $v$ has been modified by the adversary. 
A resilient query is required to correctly report the answer when no vertex in the tree path between the two vertices explicitly or implicitly defined by the query is corrupted. For example, a
$\LA(v,k)$ query correctly reports the $k$-parent $u$ of $v$ whenever every vertex in the unique path from $u$ to $v$ in $T$ is uncorrupted.

We deem our notion of resilient query to be quite natural since in any reasonable representation of $T$ the adversary can locally corrupt the parent-children relationship and hence change the observed topology of $T$. See Figure \ref{fig:resilientLA} for an example.

Our data structure occupies linear (w.r.t.\ the current number of nodes) space, and supports the $\AddLeaf$ operation and $\LA$, $\BVQ$, and $\LCA$ queries in $O(\delta)$ worst-case time. For weighted $\LA$ queries,  the above bound on the query time holds as long as $\delta = O(\mathrm{polylog} n)$.

We point out that our solution is obtained through a general vertex-coloring scheme which is, in turn, used to ``shrink'' $T$ down to a compact tree $Q$ of size $O(n/\delta)$ that can be made resilient via replication.
Each edge of $Q$ represents a path of length $\delta$ between two consecutive colored nodes in $T$.
If no corruption occurs, this coloring scheme is regular and will color all vertices having a depth that is a multiple of $\delta$. 
While it is possible for corruptions to \emph{locally} destroy the above pattern, our coloring is able to automatically recover as soon as we move away from the corrupted portions of the tree. 
We feel that such a scheme can be of independent interest as an useful tool to design other resilient data structures involving dynamic trees.

We leave the problem of understanding whether, similarly to other resilient data structures \cite{FGI09,JMM07}, one can prove a lower bound of $\Omega(\delta)$ on the time needed to perform $\AddLeaf$ operation and/or to answer our queries.

\subsection{Related work}
\label{sec:related}

\subparagraph{Non-resilient data structures.} Before discussing the known result in our faulty memory model, we fist give an overview of the closest related results in the fault-free case.
Since the landscape of data structures that answer queries on dynamic trees is vast and diverse, we will focus only on the best-known data-structures capable of answering $\LA$, $\BVQ$, or $\LCA$ queries. 

As far as $\LA$ queries are concerned, the
problem has been first formalized in \cite{BO94} and in \cite{D91}.
Both papers consider the case in which the tree $T$ is \emph{static} and 
show how to built, in linear-time, a data structure that requires linear space and that answers queries in constant worst-case time (albeit the hidden constant in \cite{BO94} is quite large).
A simple and elegant construction achieving the same (optimal) asymptotic bounds is given in \cite{BF04}.
In~\cite{D91}, the \emph{dynamic} version of the problem was also considered: the authors provide a data structure supporting both $\LA$ queries and the $\AddLeaf$ operation in constant \emph{amortized} time.
The best known \emph{dynamic} data structure is the one of \cite{AH00}, which implements the above operations in constant \emph{worst-case} time.
This data structure also supports constant-time $\BVQ$ queries and constant-time weighted $\LA$ queries when the vertex weights are polylogarithmically bounded integers.
Moreover, the solution of  \cite{AH00} also provides amortized bounds for the problem of maintaining a forest of $n$ nodes under \emph{link} operations (i.e., edge additions that connect two vertices in different trees of the forest) and $\LA$ queries. In this case, a sequence of $m$ operations requires $O(m \alpha(m,n))$ time, where $\alpha$ is the inverse Ackermann function. 

Regarding $\BVQ$ queries with integer weights, in addition to the solution discussed above (which supports leaf additions and queries in constant-time), \cite{BDR11} shows how to also support leaf deletions using $O(1)$ amortized time per update and constant worst-case query time. 

The problem of answering $\LCA$ queries is a fundamental problem which has been introduced in \cite{AHU76}. In \cite{HT84}, Harel and Tarjan show how to preprocess in linear time any \emph{static} tree in order to to build a linear-space data structure that is able to answer $\LCA$ queries in $O(1)$ time.
The case of \emph{dynamic} trees is also well-understood: it is possible to simultaneously support (i) insertions of leaves and internal nodes, (ii) deletion of leaves and internal nodes with a single child, and (iii) $\LA$ queries, in constant worst-case time per operation \cite{CH05}.

\subparagraph{Resilient data structures.} As already mentioned, the Faulty-RAM model has been introduced in \cite{FI04} and used in the context of resilient data structures 
in \cite{FGI07} where the authors focused on designing
resilient dictionaries, i.e., dynamic sets that support insertions, deletions, and lookup of elements.
Here the lookup operation is only required to answer correctly if either (i) the searched key $k$ is in the dictionary and is uncorrupted, or (ii) $k$ is not in the dictionary and no corrupted key equals $k$.
The best-known (linear-size) resilient dictionary is provided in \cite{BFFGIJMM07} and supports each operation in the optimal $O(\log n + \delta)$ worst-case time, where $n$ is the number of stored elements. 
The Faulty-RAM model has also been adopted in \cite{JMM07}, where the authors design a (linear-size) resilient priority queue, i.e., a priority queue supporting two operations: \emph{insert} (which adds a new element in the queue) and \emph{deletemin}. Here \emph{deletemin} deletes and returns either the minimum uncorrupted value or one of the corrupted values. Each operation requires $O(\log n + \delta)$ amortized time, while $\Omega(\log n + \delta)$ time is needed to answer an \emph{insert} followed by a \emph{deletemin}.

The Faulty-RAM model has also been adopted in the context of designing resilient algorithms. We refer the reader to \cite{Italiano10} for a survey on this topic. 

A resilient dictionary for a variant of the Faulty-RAM model in which the set of corruptible memory words is random (but still unknown to the algorithm) has been designed in \cite{LLM19}.

In a broader sense, problems that involve non-reliable computation have received considerable attention in the literature, especially in the context of sorting and searching. See for example \cite{Pelc02,Geissmann0LPP19,Geissmann0LP19,FeigeRPU94,Cicalese13,ChenGMS17,10.1007/978-3-030-79987-8_19}. 

\subsection{Structure of the paper}

The paper is organized as follows. Section~\ref{sc:preliminaries} introduces the used notation and formally defines the Faulty-RAM model. It also briefly describes the error-correcting replication strategy mentioned in the introduction. For technical convenience, in Section~\ref{sc:LA_static} and \ref{sc:LA_dynamic} we describe our data structure for \LA queries only. This allows us to introduce all the ideas behind the more general coloring scheme discussed above. 
As a warm up, we first consider the simpler case in which the tree $T$ is \emph{static} and is already known at construction time (Section~\ref{sc:LA_static}), and we then tackle the dynamic version of the problem (Section~\ref{sc:LA_dynamic}) for which we give our main result.
In Section~\ref{sec:handling_queries} we show how to modify our data structure to handle the other types of queries.

\section{Preliminaries}
\label{sc:preliminaries}

\subparagraph*{Notation.}
Let $T$ be a rooted tree. For each node $v \in T$, we denote with $\parent(v)$ the parent of $v$. If $\pi$ is a path, we denote by $|\pi|$ its length, i.e., the number of its edges. Given any two nodes $u,v$, we denote by $d_T(u,v)$ the length of the (unique) path between $u$ and $v$ in $T$.
Moreover, if $\pi$ traverses $u$ and $v$, we denote by $\pi[u:v]$ the subpath of $\pi$ between $u$ and $v$, endpoints included. We will use round brackets instead of square brackets to denote that the corresponding endpoint is excluded (so that, e.g., $\pi(u:v]$ denotes the subpath of $\pi$ between $u$ and $v$ where $u$ is excluded and $v$ is included).

\subparagraph*{Faulty memory model.}
We now formally describe the \emph{Faulty-RAM} model introduced by Finocchi and Italiano in \cite{FI04}. In this model the memory is divided in two regions: a \emph{safe region} with $O(1)$ memory locations, whose locations are known to the algorithm designer, and the (unreliable) \emph{main memory}.
An adaptive adversary can perform up to $\delta$ corruptions, where a corruption consists in instantly modifying the content of a word from the main memory. The adversary knows the algorithm and the current contents of the memory, has an unbounded computational power, and can simultaneously perform one or more corruptions at any point in time.
The safe region cannot be corrupted by the adversary and
there is no error-detection mechanism that allows the algorithm to distinguish the corrupted memory locations from the uncorrupted ones. 

Without assuming the existence of $O(1)$ words of safe memory, no reliable computation is possible: in particular, the  safe memory can store the code of the algorithm itself, which otherwise could be corrupted by the adversary. 

As observed in \cite{FGI09} (and already mentioned in the introduction), there is a simple strategy that allows any non-fault tolerant data structure on the RAM model to also work on the Faulty-RAM model, albeit with multiplicative $\Theta(\delta)$ blow-up in its time and space complexities.
Essentially, such a solution implements a trivial error-correcting mechanism by
simulating each memory word $w$ in the RAM model with a set $W$ of $2\delta+1$ memory words in the Faulty-RAM model: writing a value $x$ to $w$ means writing $x$ to all words in $W$, and reading $w$ means computing the majority value of the words in $W$ (which can be done in $O(\delta)$ time, and $O(1)$ space using the safe memory region and the Boyer-Moore majority vote algorithm \cite{Moore91}). We refer to such technique as the \emph{replication strategy}.

\section{Warming Up: \LA queries in Static Trees}\label{sc:LA_static}

In order to introduce our ideas, in this section we will show how to build a simplified version of our resilient data structure when the tree $T$ cannot be dynamically modified.
Our simplified data structure requires linear space and answers level-ancestor queries in $O(\delta)$ time.
As opposed to our \emph{dynamic} data structure, in this special case the tree $T$ must be known in advance and hence we need to initialize our data structure from an input tree $T$. For simplicity, we assume that no corruptions occur while our data structure is being built. Notice that this assumption can be removed by carefully using the replication strategy described in Section~\ref{sc:preliminaries}.

\subparagraph{Description of the Data Structure.}
Let $T$ be a rooted tree with $n$ nodes. To define the data structure for $T$, we need to divide the nodes of $T$ into two sets: the \emph{black} nodes and the \emph{white} nodes. We define the set of black nodes to ensure that its cardinality is $O(n/\delta)$: a node $v$ in $T$ is \emph{black} if we simultaneously have that (i) its depth in $T$ is a multiple of $\delta$, and (ii) the subtree of $T$ rooted in $v$ has height at least $\delta-1$. A node $v$ in $T$ is \emph{white} if it is not black.
We notice that for each black $v$ node in $T$ there are at least $\delta$ distinct nodes (i.e., all the vertices in the path from $v$ to any vertex having depth $\delta-1$ in the subtree of $T$ rooted at $v$), thus implying that the total number of black nodes in $T$ is at most $n/\delta$.

If we define a relation of parenthood for the black nodes of $T$, we can define a new \emph{black tree} $Q$ in which each vertex $\overline{v}$ is associated with black vertex $v$ of $T$. 
The parent of $\overline{v}$ in $Q$ is the vertex $\overline{u}$ corresponding to the lowest black proper ancestor $u$ of $v$ in $T$. See Figure~\ref{fig:TandQ_staticquery} for an example.

Our data structure stores the (colored) tree $T$, as described in the following, along with an additional data structure $D_Q$ that is able to answer \LA queries on $Q$.
The tree $T$ is stored as an array of records, where each record is associated with a vertex of $T$, occupies $\Theta(\log n)$ bits, and is stored in a single memory word.
The memory word associated with a node $v$ stores:
\begin{itemize}
\item a pointer $p_v$ to $\parent(v)$, if any. If $v$ is the root of $T$ then $p_v = \mynull$;
\item a pointer $q_v$ to the corresponding node $\overline{v}$ in $Q$, if any. If no such node exists, i.e. if $v$ is white, then $q_v = \mynull$.
\end{itemize}

\begin{algorithm}[t]
\caption{Answers a level ancestor query $\LA(v,k)$ in the special case of static trees.}
\label{alg:static-query}
\small

\lIf{$k \leq 2 \delta$}
{
	\Return $\texttt{climb}(v,k)$\label{ln:climb-littlek}
}

\BlankLine

Climb up the tree $T$ from $v$ for $2 \delta$ nodes searching a black node\; \label{ln:findblacknode}
\lIf{the previous procedure did not find a black node}
{
	\Return error%
}

\BlankLine

$v' \gets$ a black node found in the previous procedure\;
$d \gets$ distance between $v'$ and $v$; \, $k' \gets k -d$\;
$u' \gets \LA_Q(q_{v'},\lfloor k'/\delta \rfloor)$\label{ln:query_Q}\; 
$k_{\text{rest}} \gets k' -\lfloor k'/\delta\rfloor\cdot \delta$\;
\Return $\texttt{climb}(u',k_{\text{rest}})$\;\label{ln:climb-krest}
\end{algorithm}

Moreover we maintain, for each vertex $\overline{v}$ of $Q$, a pointer to the corresponding vertex $v$ of $T$ as satellite data. The data structure $D_Q$ is the resilient version of any (non-resilient) data structure that is capable of answering $\LA$ queries on static trees in constant time and requires linear space (see, e.g., the data structure in \cite{BF04}). 

As we observed before, any data structure can be made resilient with a multiplicative $\Theta(\delta)$ blow-up in its time and space complexities.
In our case, since the number of vertices in $Q$ is $O(n/\delta)$, the 
final space required to store $D_Q$ is $O(n)$ and the query time becomes $O(\delta)$.
Notice that, in spite of the (at most $\delta$) memory corruptions performed by the adversary, the data structure $D_Q$ always returns the correct answer to all possible $\LA$ queries on  $Q$.
We will denote by $\LA_Q(\overline{v},k)$ the level-ancestor query on $Q$, which returns the vertex of $T$ corresponding to the $k$-parent of $v$ in $Q$.

\subparagraph{The resilient level-ancestor query.}

In this section we show  how to implement our resilient \LA query.
We start by defining a routine that will be useful in the sequel: if $v$ is a node of $T$ and $i$ is a non-negative integer, we denote by $\climb(v,i)$ a procedure that returns the vertex reached by a walk on $T$ that starts from $v$ and iteratively moves to the parent of the current vertex $i$ times. When the procedure encounters a vertex $u$ with pointer $p_u=\mynull$ that has to be followed, $\texttt{climb}(v,i)$ reports that the root has been reached.  
Notice that $\texttt{climb}(v,i)$ requires $O(i)$ time and, whenever no corrupted vertices are encountered during the walk, it correctly returns the $i$-parent of $v$. 
Although the $\texttt{climb}(\cdot, \cdot)$ procedure could immediately be used to answer an $\LA$ query, doing so require $\Omega(n)$ time in the worst case. To improve the query time we use the data structure $D_Q$ described above and we distinguish between \emph{short} and \emph{long} $\LA(v,k)$ queries depending on the value of $k$.

Short queries, i.e., queries $\LA(v,k)$ with $k \le 2\delta$, are handled
by simply invoking $\texttt{climb}(v,k)$ and, from the above observation, it follows that this is a resilient query.
For longer queries the idea is that of locating a nearby black ancestor of $v$, performing an $\LA_Q$ query on $Q$ to quickly reach a black descendant $u'$ of the $k$-parent $w$ of $v$ such that $d(u', w) \le \delta$, and finally using the \climb \ procedure once more to reach $w$ from $u'$. See Algorithm~\ref{alg:static-query}.

During the execution of our resilient query algorithm we always ensure that all followed pointers are valid. Since we are dealing with a static tree $T$, we can handle invalid pointers simply by halting the whole query procedure and reporting an error. A slightly more sophisticated handling of invalid pointers will be used to tackle the dynamic case. An example $\LA$ query is given in Figure~\ref{fig:TandQ_staticquery}.

The correctness of the above algorithm immediately follows from the fact that, when no vertex between $v$ and the $k$-th ancestor $v$ is corrupted, $v$ must have a black ancestor at distance at most $2 \delta$ and from the fact that the replication strategy ensures that all queries on $Q$ are always answered correctly.\footnote{Here the distance of $2\delta$ is essentially tight as it can be seen, e.g., by considering a tree $T$ consisting of a path of length $2\delta-2$ rooted in one of its endpoints. The only black vertex of $T$ is the root. Notice how the vertex $u$ at depth $\delta$ is white since the subtree of $T$ rooted in $u$ has height $\delta-2$.}

\begin{figure}
    \centering
    \includegraphics[scale=.75]{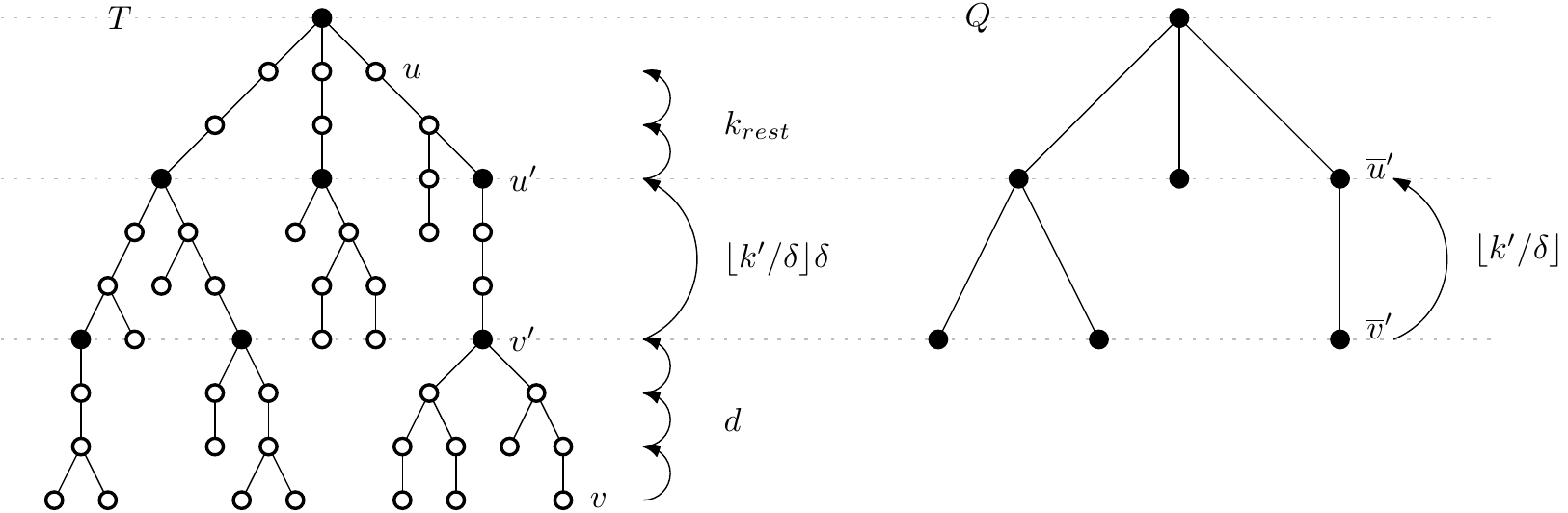}
    \caption{Left: A static tree $T$ that has been colored according to the scheme in Section~\ref{sc:LA_static} for $\delta=3$. Right: the corresponding black tree $Q$.
    We also show the path climbed while answering the query $\LA(v,k)$ with $k=8$. In this case $d=3$, $\lfloor k'/\delta\rfloor=1$ and $k_{\text{rest}}=2$.
    Notice how $Q$ is used to quickly reach $u'$ from $v'$.}
    \label{fig:TandQ_staticquery}
\end{figure}

To show that Algorithm~\ref{alg:static-query} answers an $\LA$ query in $O(\delta)$ time, we notice that the $\climb$ operations in lines \ref{ln:climb-littlek} and \ref{ln:climb-krest} require time $O(\delta)$, and so does line~\ref{ln:findblacknode}. Moreover, the query to $D_Q$ (line~\ref{ln:query_Q}) can also be performed in $O(\delta)$ time as discussed above.

\section{\LA queries in dynamic trees}\label{sc:LA_dynamic}

In this section we provide our main result for $\LA$ queries. In Section~\ref{sec:handling_queries}, we show how our ideas can be extended to also handle weighted $\LA$, $\BVQ$, and $\LCA$ queries.

\subsection{Description of the Data Structure}
\label{ssec:description_dynamic_datastructure}

Some of the key ideas behind our data structure for $\LA$ queries in dynamic trees are extensions of the ones used for static case. Namely, the nodes of $T$ are colored with either \emph{black} or \emph{white}, the set of \emph{black} nodes will have size $O(n/\delta)$, and it will correspond to the vertex set of an auxiliary \emph{black forest} $Q$. Ideally, in absence of corruptions, $Q$ is exactly the black tree as defined in the static case, namely the tree in which the parent of each (black) node $\overline{v}$ in $Q$ is the vertex $\overline{u}$ associated with the lowest black proper ancestor $u$ of $v$ in $T$. 

Moreover, we would still like the vertices of $T$ having a depth that is a multiple of $\delta$ to be colored black, similarly to the static case.
However, we can no longer afford to maintain such a rigid coloring scheme since the tree is now being dynamically constructed via successive $\AddLeaf$ operations, and the adversary's corruptions might cause vertices to become miscolored.
We will however ensure that such a regular coloring pattern will be followed by the portions of $T$ that are sufficiently distant from the adversary's corruptions.
This will allow us to answer $\LA$ queries using a strategy similar to the one employed for the static case.

Our data structure stores the following information. The record of a node $v$ maintains, in addition to the pointer $p_v$ to its parent and to the pointer $q_v$ to the corresponding node $\overline{v}$ in $Q$ (if any), an additional field $\nodeflag_v$. Intuitively, $\nodeflag_v$ can be thought of as a Boolean value in $\{\bot, \top\}$. The initial value of a flag is $\bot$ and we say that the flag is \emph{unspent}. Spending a flag means setting it to $\top$.
We will spend these flags to ``pay'' for the creation of new black nodes. Spent flags will also signal the presence of a nearby black ancestor.

For technical reasons, we also allow an unspent flag $\nodeflag_v$ to be additionally \emph{annotated} with a pair $(x,i)$ where $x$ is (the name of) a node and $i$ is an integer.
In practice this amounts to setting $\nodeflag_v$ to $(x,i)$, which is logically interpreted as $\bot$. Such an annotated flag is still unspent.
This provides an additional safeguard against corruptions that may occur during the execution of our leaf insertion algorithm (see Section~\ref{sec:addleaf}).

The node records are stored into a dynamic array $\A$, whose current size $n$ is kept in the safe region of memory. This array supports both elements insertions and random accesses in constant worst-case time.\footnote{The standard textbook technique which
handles insertions into already full array by moving the current elements into a new array of double capacity already achieves $O(1)$ amortized time per insertion. With some additional technical care, the above bound also holds in the worst case.
The idea is to distribute the work needed to move elements into the new array over the insertions operations that would cause the current array to become full (it suffices to move 2 elements per insertion). Using this scheme, at any point in time, each element is stored into a single memory word.}

The pointer $p_v$ is then the index (in $\{0, \dots, n-1\}$) of the record corresponding to the parent of $v$ in $\A$.
Initially, $\A$ only contains the root $r$ of $T$ at index $0$. Moreover, we will always store new leaves at the end of $\A$ so that, in absence of corruptions, the index of a vertex $v$ in $\A$ is always smaller than the index of any of its descendants.
As a consequence, whenever we observe the index stored in pointer $p_v$ is greater than or equal to the index of $v$ itself, we know that $v$ must have been corrupted by the adversary.
We find convenient to use the above fact to simplify the handling of corrupted vertices: whenever we encounter an invalid pointer $p_v \ge v$ we treat it as being equal to $0$, i.e., an invalid pointer $p_v$ always refers to the root $r$ of $T$. This rules also applies to any read pointer, including those accessed by the $\texttt{climb}(\cdot, \cdot)$ procedure already defined in Section~\ref{sc:LA_static}.

Then the (possibly corrupted) contents of $\A$, at any point in time, induce a \emph{noisy tree} $\pseudoT$ whose root is $r$, and the parent of each vertex $v \neq r$ is the vertex pointed by $p_v$ according to the above rule.
Clearly, if no corruptions occur $T$ and $\pseudoT$ coincide.

Moreover, we store a resilient data structure $D_Q$ that, in addition to the already-defined $\LA_Q(\overline{v}, k)$ query, also supports the following additional operations in $O(\delta)$ time.
\begin{description}
\item[$\NewTree_Q(v)$:] Given a vertex $v$ of $T$, it creates a new tree in the forest $Q$ consisting of a single vertex $\overline{v}$ associated to $v$, and it returns a pointer to $\overline{v}$.
\item[$\AddLeaf_Q(\overline{u}, v)$:] Given a vertex $\overline{u}$ of $Q$, and a vertex $v$ of $T$, it creates a new vertex $\overline{v}$ associated to $v$ as a children of $\overline{u}$ in $Q$. Finally, it returns a pointer to the newly added vertex $\overline{v}$.
\end{description}

This data structure $D_Q$ is the resilient version, obtained using the replication strategy, of the linear-size data structure that supports both the $\AddLeaf$ operation and $\LA$ queries in constant time \cite{AH00}.
Notice that $D_Q$ always returns the correct answer to all possible $\LA$ queries on $Q$. Moreover, once we ensure that the number of vertices that become black (and hence the size of $Q$) is always $O(n/\delta)$, we have that  the (resilient) data structure $D_Q$ requires $O(n)$ space (this will be shown formally in the proof Theorem~\ref{thm:main_LA}).

\subsection{The \texttt{AddLeaf} operation}
\label{sec:addleaf}

Before describing our implementation of the $\AddLeaf$ operation, it is useful to give some additional definitions.
We say that $v$ is \emph{near-a-black} in a tree $\Tilde{T}$ if there exists some $k \in \{1, 2, \dots, \delta\}$ such that the $k$-parent of $v$ in $\Tilde{T}$ is black.
Moreover, we say that $v$ is \emph{black-free} in $\Tilde{T}$ if no $k$-parent of $v$ in $\Tilde{T}$ for $k \in \{1, 2, \dots, 2\delta-1\}$ is black.

The procedure $\AddLeaf(x_{par})$ takes a vertex $x_{par}$ of $T$ as input and adds a new child $x$ of $x_{par}$ to $T$ (see Algorithm~\ref{alg:add-leaf}).
The record corresponding to new vertex $x$ is appended at the end of the dynamic array $\A$. For simplicity we will assume that, during the execution of $\AddLeaf(x_{par})$, the record of vertex $x$ is never corrupted by the adversary. This can be guaranteed without loss of generality since a (temporary) record for $x$ can be kept in safe memory and copied back to $\A$ (which is stored in the unreliable main memory) at the end of the procedure.

Our algorithm consists of a first  \emph{discovery} phase and possibly of a second additional \emph{execution} phase.
The aim of the discovery phase is that of exploring the current tree by climbing up to $3\delta-1$ levels of $\pseudoT$ from $x$ while gathering information for the second phase.
In order to do so, Algorithm~\ref{alg:add-leaf} climbs $\delta$ levels of $\pseudoT$ from the newly inserted node $x$, reaching a vertex $y$, and checks during the process that all the flags associated with the traversed nodes are unspent. 
If any of these flags is spent, we immediately return from the $\AddLeaf(x_{par})$ procedure without performing the execution phase. Otherwise, the algorithm climbs $2\delta-1$ further levels from $y$ to determine whether $y$ appears to be black-free or near-a-black.
In the latter case, it keeps track of the distance $\ell$ from $y$ to the closest black proper ancestor $y'$ of $y$ that is encountered.  
If $y$ is neither black-free nor near-a-black, we return from the $\AddLeaf(x_{par})$ procedure (without performing the execution phase),
otherwise we move on to the execution phase.
A technical detail of the discovery phase is the following: while climbing from $x_{par}$ to $y$, the generic $i$-th unspent flag is annotated with $(x, i)$ (possibly overwriting any existing previous annotation) and will be checked by the execution phase. Recall that these flags remain unspent.

The execution phase once again climbs $\delta$ levels of $\pseudoT$ staring from $x$, with the goal of changing the color of an existing white vertex to black (hence creating a corresponding black node in $Q$). 
This is guaranteed to happen unless the annotations of the unspent flags set during the discovery phase reveal that one such vertex has been corrupted in the meantime.  
The creation of a new black vertex in $Q$ is ``paid for'' by \emph{spending} these $\delta$ unspent flags (i.e., setting them to $\top$).
The position of the new black vertex depends on whether $y$ was near-a-black or black-free.
In the former case the vertex $y'$ discovered in the first phase will be the $\delta$-parent of the new black vertex $x'$, and a new leaf $\overline{x}'$ is appended to $\overline{y}'$ in $Q$.
In the latter case, $y$ will become black and a new tree containing a single vertex $\overline{y}$ is added to $Q$.
Notice that, if a vertex $b$ is colored black during the $\AddLeaf$ operation, the execution phase always spends $\nodeflag_b$.

\begin{algorithm}[t]
\SetKw{Break}{break}
\SetKw{And}{and}

\caption{$\AddLeaf(x_{par})$}
\label{alg:add-leaf}
\small
Add a new record $x$ at the end of $\A$; \,
$p_x \gets x_{par}$; \, $\nodeflag_x \gets \bot$; \, $q_x \gets \mynull$\; 

\BlankLine
\tcp{Discovery Phase}
$y \gets x$\tcp*{Check the flags of the lowest $\delta$ proper ancestors of $x$}
\label{ln:phase1-flag1}
\For{$i=1,\dots, \delta$}{
    \lIf{$y=r$}{\Return $x$}
    $y \gets p_y$\;
    \lIf(\tcp*[f]{Return immediately if a spent $\nodeflag$ is found}){$\nodeflag_y = \top$}{\Return $x$}
    $\nodeflag_y \gets (x,i)$\tcp*{Annotate $\nodeflag_y$}
}\label{ln:phase1-flag2}

\BlankLine
\tcp{Check whether $y$ is near-a-black.} 
\tcp{\mbox{$\ell$ will be the distance to the closest proper  black ancestor $y'$ of $y$, if any}}
$y' \gets y$; \quad $\ell \gets 0$; \quad $\nearblack \gets \texttt{false}$\;\label{ln:superinsertion}
\While{$\ell < \delta$ \And  $y' \neq r$ \And near\_black$=\texttt{false}$}
{
    $y' \gets p_{y'}$; \quad $\ell \gets \ell+1$\;
    \lIf{$y'$ is black}{near\_black $\gets \texttt{true}$}
}

\tcp{If $y$ is not near-a-black, check whether it is black-free}\label{ln:discovery-near-a-black}
\If{$\nearblack= \texttt{false}$}
{
    $z \gets y'$\;
    \For{$\delta-1$ times}
    {
        \lIf{$z=r$}{\Break}
        $z \gets p_z$\;
        \lIf{$z$ is black}{\Return $x$}
    }   
}

\BlankLine
\tcp{Execution Phase. (Node $y$ was either near-a-black or black-free)}
\tcp{Acquire the flags of the lowest $\delta$ proper ancestors of $x$}
$z \gets x$\;
\label{ln:phase2-flag1}
\For{$i=1, \dots, \delta$}
{
    \lIf{$z=r$}{\Return $x$}
    $z \gets p_z$\;
    \lIf(\tcp*[f]{Check the annotation of $\nodeflag_z$}){$\nodeflag_z \neq (x,i)$}{\Return $x$}
    \lIf(\tcp*[f]{$y'$ is the $\delta$-parent of $x'$}){$\nearblack=\texttt{true}$ \And $i=\ell$}{$x' \gets z$}
    \label{ln:execution-reachingx'}
    $\nodeflag_z \gets \top$\tcp*{Spend $\nodeflag_z$}
}\label{ln:phase2-flag2}

\BlankLine
\If{$\nearblack=\texttt{true}$}
{
    $q_{x'} \gets \AddLeaf_Q(q_{y'},x')$\;\label{ln:add-leafQ-nearablack}
}
\Else
{
    $q_y \gets \NewTree_Q(y)$\;
}
\Return $x$\;
\end{algorithm}

\subsection{Analysis of the data structure}

In this section we analyze our data structure. The core of the analysis is to show that the $\AddLeaf$ operation in Algorithm \ref{alg:add-leaf} guarantees that in $\pseudoT$, if we are sufficiently distant from all the corrupted vertices, the black nodes are regularly distributed. The formal property is stated in Lemma \ref{lem:blacknodes_pattern}. We first need to prove auxiliary properties.  
In Figure \ref{fig:example_corruptions} we give an example that shows that, even in an uncorrupted path, if we are not sufficiently distant from corruptions, the black nodes can form irregular patterns in the path.

\begin{figure}[t]
    \centering
    \includegraphics[width=\textwidth]{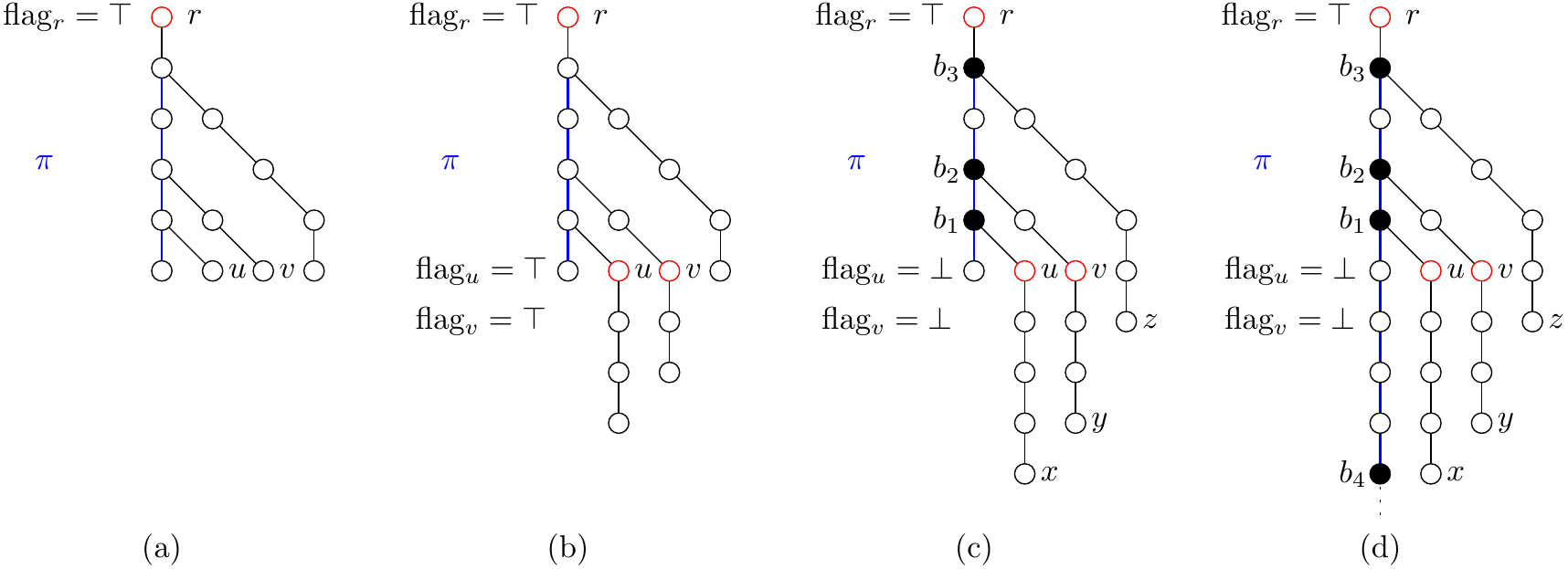}
    \caption{An example showing that an uncorrupted path $\pi$ (depicted in blue) can exhibit an irregular pattern of black vertices (d). 
    Situation (a) can be reached when the adversary corrupts $r$ by setting $\nodeflag_r = \top$ before the insertions of the other nodes take place.
    To obtain (b), the adversary can set $\nodeflag_u = \nodeflag_v = \top$, thus corrupting $u$ and $v$ before $u$ and $v$'s descendants are inserted.
    If the adversary sets $\nodeflag_u$ and $\nodeflag_v$ back to $\bot$ before $x$, $y$, and $z$ are inserted (in this order), we arrive in configuration (c) in which $b_1$, $b_2$, and $b_3$ have been colored black.
    Inserting the remaining vertices yields (d).}
    \label{fig:example_corruptions}
\end{figure}

The following lemma shows that if the flag of a vertex $w$ appears to be spent, then either there must be a nearby black ancestor of $w$, unless a nearby corruption occurred. See Figure~\ref{fig:main-merge}~(a).

\begin{lemma}
Let $w$ and $z$ be two nodes such that $z$ is the $\delta$-parent of $w$ in $T$ and such that no node in the path $\pi$ from $z$ to $w$ in $T$ has been corrupted. 
If $\nodeflag_w = \top$, then there exists a black node in $\pi(z:w]$.
\label{lem:preliminary}
\end{lemma}
\begin{proof}
Let $x$ be the node whose insertion in $T$ caused $\nodeflag_w$ to be set to $\top$. Moreover, let $P$ be the path of length $\delta$ from $x$ to $y$ traversed in the discovery phase of Algorithm~\ref{alg:add-leaf} in lines \ref{ln:phase1-flag1}--\ref{ln:phase1-flag2}.
Similarly, let $P'$ be the path from $x$ to $y$ traversed in the execution phase of Algorithm~\ref{alg:add-leaf} in lines \ref{ln:phase2-flag1}--\ref{ln:phase2-flag2}.

Clearly, $P'$ contains $w$. Moreover, if $w$ is the $i$-th node traversed in $P'$, then $\nodeflag_w = (x, i)$ in the execution phase and (since $w$ is uncorrupted), $\nodeflag_w$ was set to $(x, i)$ in the discovery phase.
As a consequence, $w$ is also the $i$-th node in $P$ and $P[y:w] = P'[y:w]$.
Hence, $y$ is at distance at most $\delta-1$ from $w$ in $P$ (and in $T$) showing that $z$ is a proper ancestor of $y$.
Therefore all nodes in $P'[y:w]$ are uncorrupted, and the loop in  in lines \ref{ln:phase2-flag1}--\ref{ln:phase2-flag2} of Algorithm~\ref{alg:add-leaf} is executed to completion.
This ensures that the execution phase will color a node $b$ black. We distinguish two cases depending on whether $y$ was observed to be near-a-black or black-free in the discovery phase.

If $y$ is black-free, then $b$ is exactly $y$ and the claim follows.
Otherwise, $y$ is near-a-black and the discovery phase computed the distance $\ell$ between $x$ and its closest black proper ancestor.
If $\ell \ge i$, then Algorithm~\ref{alg:add-leaf} colors a vertex in $P(z:w]=\pi(z:w]$ black. Otherwise, if $\ell < i$, the discovery phase observed that the $\ell$-parent $y'$ of $y$ was black. Since $\ell < \delta$, $y'$ lies in $\pi(z:y]$.
\end{proof}

\noindent Next lemma shows that an uncorrupted path of length at least $3\delta$ must contain a black vertex.

\begin{lemma}
Let $x$ and $z$ be nodes in $T$ such that $z$ is the $3\delta$-parent of $x$ in $T$ and such that no node in the path $\pi$ from $x$ to $z$ in $T$ has been corrupted.
Then, there exists a black node $w$ in $\pi[z:x)$. 
\label{lem:invariant1}
\end{lemma}
\begin{proof}
Since no vertex in $\pi$ has been corrupted, the path $\pi$ must also belong to the noisy tree $\pseudoT$.
In the rest of the proof we assume that $\pi[z:x)$ contains no black nodes and show that this leads to a contradiction.

Let $y$ the $\delta$-parent of $x$ in $\pi$ and let $t_x$ be the time at which the $\AddLeaf(\cdot)$ operation that adds $x$ to $T$ is invoked. 
We know that, at time $t_x$, there exists no node $w$ in $\pi[y:x)$ such that $\nodeflag_w=\top$ since otherwise Lemma \ref{lem:preliminary} would immediately imply the existence of a black node in $\pi[z:w]$ contradicting the initial assumption. Then, the invocation of Algorithm~\ref{alg:add-leaf} that inserts $x$ also performs its execution phase.

Moreover, $y$ must be black-free at time $t_x$, and hence it is colored black during such a phase (refer to the pseudocode of Algorithm~\ref{alg:add-leaf}, and recall that a black-free node is not near-a-black). Since $y$ is not corrupted it must still be black, leading to a contradiction.
\end{proof}

Recall that we would like each uncorrupted path to contain 
a black vertex every $\delta$ levels. Consider an uncorrupted path $\pi$ of length between $\delta$ and $2\delta$ with a single black vertex $z$ on top. Then, the vertex at distance $\delta$ from $z$ is ``overdue'' to become black. Next lemma shows that all flags associated with descendants of the overdue vertex in $\pi$ must be unspent. In some sense, the data structure is preparing to recolor the missing black vertex. This will happen once $\delta$ unspent flags are available. See Figure~\ref{fig:main-merge}~(b).

\begin{figure}[t]
    \centering
    \includegraphics[width=\textwidth]{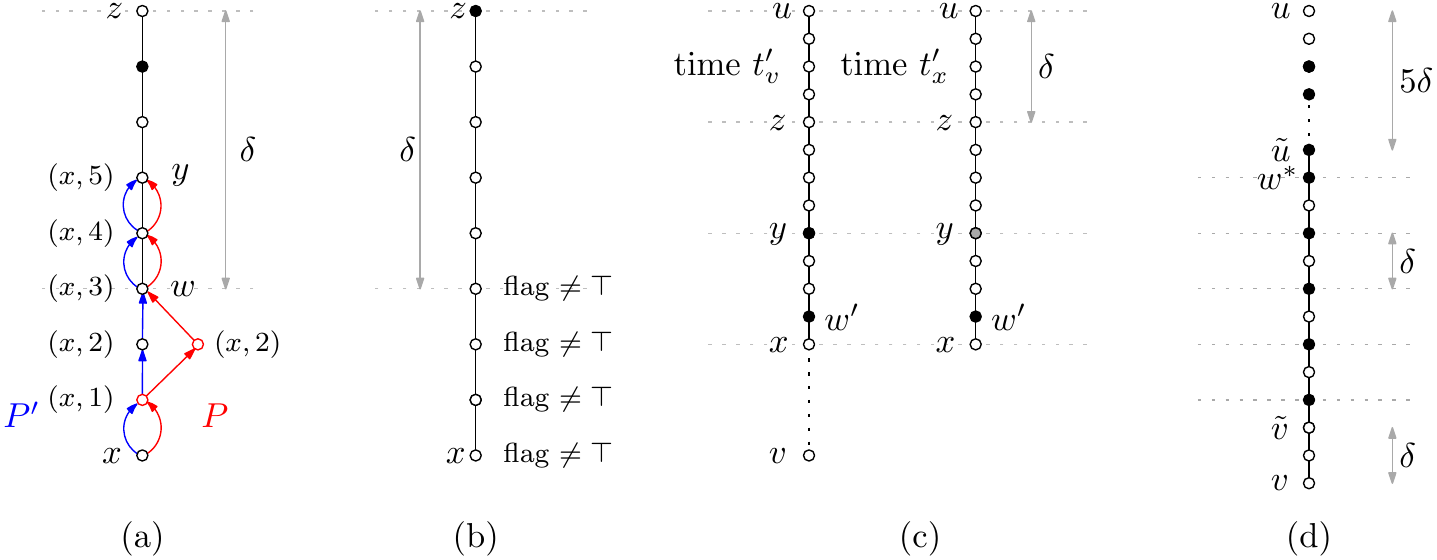}
    \caption{(a) Graphical representation of the proof of Lemma~\ref{lem:preliminary}, for $\delta=5$. (b), (c), (d): Representations of the statements of Lemma~\ref{lem:invariant3}, Lemma~\ref{lem:invariant4}, and Lemma~\ref{lem:blacknodes_pattern}, respectively.}
    \label{fig:main-merge}
\end{figure}

\begin{lemma}
\label{lem:invariant3}
Let $x$ and $z$ be two nodes in $T$ such that: $z$ is an ancestor of $x$ in $T$, no node in the path $\pi$ from $z$ to $x$ in $T$ has been corrupted, and $\delta \leq |\pi| < 2\delta$.
We have that, immediately after vertex $x$ is inserted, if the only black vertex in $\pi$ is $z$ then all the nodes $w$ in $\pi$ at distance at least $\delta$ from $z$ in $T$ are such that $\nodeflag_w \neq \top$.
\end{lemma}
\begin{proof}
Since no vertex in $\pi$ has been corrupted, the path $\pi$ must also belong to the noisy tree $\pseudoT$.
In what follows, we prove that, immediately after vertex $x$ is inserted, the existence of a node $w$ between $x$ and $z$ in $\pi$ such that $d_{\pseudoT}(w,z) \geq \delta$ and $\nodeflag_w=\top$ leads to a contradiction.
Indeed, since $\nodeflag_w = \top$, Lemma \ref{lem:preliminary} implies the existence of a black node in $\pi(z:w]$, and this contradicts the fact that $z$ is the only black node in $\pi[z:x]$.
\end{proof}

The next technical lemma is about the timing at which vertices of a long uncorrupted path become black. This will be instrumental to prove Lemma~\ref{lem:blacknodes_pattern}.
See Figure \ref{fig:main-merge}~(c).
\begin{lemma}
\label{lem:invariant4}
Let $u$ and $v$ be two nodes in $T$ such that $u$ is an ancestor of $v$, $d_T(u,v) \geq 3\delta$ and no node in the path $\pi$ from $v$ to $u$ in $T$ has been corrupted.
Let $y$ (resp. $x$) be the node in $\pi$ at distance $2\delta$ (resp. $3\delta$) from $u$ in $\pi$. 
Let $t'_v$ (resp. $t'_x$)  be the time immediately after the vertex $v$ (resp. $x$) is inserted. If the node $y$ is black at time $t'_v$, then there exists a node $w'$ in $\pi[y:x]$ that is black at time $t'_x$.
\end{lemma}
\begin{proof}
Since no vertex in $\pi$ has been corrupted, the path $\pi$ must also belong to the noisy tree $\pseudoT$. In the rest of the proof we assume towards a contradiction that $y$ is black at time $t'_v$, yet there are no black nodes in $\pi[y:x]$ at time $t'_x$.

Let $z$ be the $\delta$-parent of $y$ in $\pi$.
Let $\overline{t}_y$ be the time immediately before $y$ is colored black. At time $\overline{t}_y$ there are only two possible scenarios: 
\begin{description}
    \item[Scenario 1:] At time $\overline{t}_y$, the node $y$ is black-free; 
    \item[Scenario 2:] At time $\overline{t}_y$, the node $z$ is the only black node in $T$ in $\pi[z:y]$.
\end{description}

We denote with $t_x$ the time immediately before vertex $x$ is inserted in $T$ and we consider the two scenarios separately (notice that $\overline{t}_y$ refers to a later time than $t_x$). We split scenario 1 into  two additional subcases:
\begin{description}
    \item[Subcase 1.1:] at time $t_x$ all the nodes $w$ in $\pi[y:x)$ are such that $\nodeflag_w\neq \top$;
    \item[Subcase 1.2:] at time $t_x$ there is a node $w$ in $\pi[y:x)$ such that $\nodeflag_w=\top$.
\end{description}

We start considering subcase 1.1. 
Since $\overline{t}_y$ follows $t_x$, and $y$ is black-free at time $\overline{t}_y$, vertex $y$ must also be black-free at time $t_x$.
Then, during the insertion of $x$, Algorithm~\ref{alg:add-leaf} colors $y$ black yielding a contradiction.

We now analyze subcase 1.2. Since $\nodeflag_w = \top$, Lemma \ref{lem:preliminary} implies the existence of a black node $b$ in $\pi[w,z)$ and, since we assume that there are no black nodes in $\pi[y:x]$, $b$ is in $\pi(z:y)$.
This shows that $y$ cannot be black-free at time $\overline{t}_y$ and contradicts the hypothesis of scenario 1.1.

We now consider Scenario 2, which we subdivide into three subcases:
\begin{description} 
\item[Subcase 2.1.] at time $t_x$ all the nodes $w$ in $\pi[y:x)$ are such that $\nodeflag_w\neq \top$ and $z$ is white;
\item[Subcase 2.2.] at time $t_x$ all the nodes $w$ in $\pi[y:x)$ are such that $\nodeflag_w\neq \top$ and $z$ is black;
\item[Subcase 2.3.] at time $t_x$ there is a node $w$ in $\pi[y:x)$ such that $\nodeflag_w=\top$.
\end{description}

We start by handling subcase 2.1. For the initial assumption, and for definition of this case, we have that there are no black nodes in $\pi[z:x]$ at time $t_x$. Since $z$ is colored black at some time $\overline{t}_z$ following $t_x$, we know that the $\delta-1$ nodes ancestor of $z$ are not black at time $t_x$, since this is incompatible with the fact that $z$ will become black. Since $\pi$ is not corrupted, we know that $y$ is black-free in $T$ at time $t_x$. This implies that $y$ is colored black during the insertion of $x$ in $T$, and hence $y$ is black at time $t'_x$ contradicting our hypotheses.

We proceed by analyzing subcase 2.2. 
At time $\overline{t}_y$ all nodes in $\pi[z:y]$, except for $z$, are white and hence the same is true at time $t_x$.
Since $z$ is black at time $t_x$ and $\nodeflag_w \neq \top$ for all nodes $w$ in $\pi[y:x)$, the \AddLeaf procedure adding $x$ will color $y$ black. Hence $y$ is black at time $t'_x$. This is a contradiction.

We now consider subcase 2.3. Together with Lemma \ref{lem:preliminary}, $\nodeflag_w = \top$ implies the existence of a black node $b$ in $\pi(z:w]$. Since we assume all the nodes in $\pi[y:x]$ to be white, the black node $b$ is in $\pi(z,y)$, contradicting the hypothesis of scenario 2.
\end{proof}

Now, we are ready to prove our main property about the pattern of black vertices discussed at the beginning of this section. See Figure \ref{fig:main-merge} (d).

\begin{lemma}
\label{lem:blacknodes_pattern}
Let $u$ and $v$ be two nodes in $T$ such that $u$ is an ancestor of $v$, the distance between $u$ and $v$ is at least $7\delta$, and no node in the path $\pi$ from $u$ to $v$ has been corrupted.
Let $\Tilde{u}$ be the node at distance $5\delta$ from $u$ in $\pi$ and let $\Tilde{v}$ be the node at distance $\delta$ from $v$ in $\pi$.
Then there is a black node $w^*$ in $\pi[\Tilde{u}:v]$ such that:
\begin{itemize}
    \item The distance between $w^*$ and $\Tilde{u}$ is at most $\delta$.
    \item A generic node in $\pi[w^* : \Tilde{v}]$ at distance $d$ from $w^*$  is black iff $d$ is a multiple of $\delta$.
    Moreover, if  $w$ is a black vertex in $\pi(w^*,\Tilde{v}]$ and $\overline{w}$ is the associated black vertex in $Q$, the parent of $\overline{w}$ in $Q$ corresponds to the $\delta$-parent of $w$ in $\pi$.
\end{itemize}
\end{lemma}
\begin{proof} 
Since no vertex in $\pi$ has been corrupted, the path $\pi$ must also belong to the noisy tree $\pseudoT$.
Then, Lemma~\ref{lem:invariant1} ensures that, at any time following the insertion of $\Tilde{u}$ in $T$, there exists a black ancestor $y$ of $\Tilde{u}$ such that $d_\pi(y,\Tilde{u}) \le 3\delta$.
Such a vertex $y$ is the $\delta$-parent of some vertex $x$ in $\pi$. We denote by $u'$ the $2\delta$-parent of $y$ in $\pi$ and by $t'_x$ the time immediately after $x$ is inserted.
Since the length of $\pi[u':v]$ is at least $3\delta$ and $y$ must be black when $v$ is inserted, we can invoke Lemma~\ref{lem:invariant4} to conclude that there exists a node in $\pi[y:x]$ that is black at time $t'_x$. 
We choose $w_0$ as the closest ancestor of $x$ that is black at time $t'_x$.
Moreover, for $i = 1, \dots, \left\lfloor |\pi[w_0 : v]| / \delta \right\rfloor$
we let $w_i$ be the unique vertex at distance $\delta i$ from $w_0$ in $\pi[w_0, v]$.
Finally, let $t'_i$ be the time immediately after the insertion of $w_i$ into $T$.

We will prove by induction on $i \ge 1$ that (i) at time $t'_i$, all vertices $w_0, w_1, \dots, w_{i-1}$ are black; (ii)  from time $t'_i$ onward, all vertices in $\pi[w_0, w_i]$  that do not belong to  $\{ w_0, w_1, \dots, w_i \}$ are white.

We start by considering the base case $i=1$.
Regarding (i), we know that $w_0$ is black at time $t'_x$, and hence $w_0$ is also black at time $t'_1$ (which cannot precede $t'_x$). Regarding (ii), by our choice of $w_0$ we know that at time $t'_x$, the only black vertex in $\pi[w_0, x]$ is $w_0$.
Moreover, Algorithm~\ref{alg:add-leaf} can only color a node $b$ black if none of the $\delta-1$ lowest proper ancestors of $b$ is black. This implies that no vertex in $\pi(w_0, w_1)$ will be colored black. 

We now assume that the claim is true up to $i \ge 1$ and prove it for $i+1$.
We first argue that the following property holds: (*) at time $t'_{i+1}$ all vertices in $\pi(w_i:w_{i+1})$ are white.
Indeed, suppose towards a contradiction that there exists some black  vertex $b$ in $\pi(w_i:w_{i+1})$ at time $t'_{i+1}$.
When $b$ was colored black, either its $\delta$-parent $b'$ was black or $b$ was black-free. In the former case we immediately have a contradiction since $b'$ must be a vertex of $\pi(w_{i-1}, w_i)$ but all such vertices are white by the induction hypothesis. In the latter case $b$ must have been colored black after the insertion of $w_i$ but, by the induction hypothesis, we know that from time $t'_i$ onwards $w_{i-1}$ is black. This contradicts the hypothesis that $b$ was black-free.

Next, we prove (i). Suppose towards a contradiction that $w_i$ is white at time $t'_{i+1}$.
Then, using (*) and the induction hypothesis, we can invoke Lemma~\ref{lem:invariant3} on
the subpath of $\pi$ between $w_{i-1}$ and the parent of $w_{i+1}$
to conclude that all nodes $w$ in $\pi[w_i:w_{i+1})$ are such that $\nodeflag_w \neq \top$.
Hence, during the insertion of $w_{i+1}$, Algorithm~\ref{alg:add-leaf} reaches line~\ref{ln:superinsertion} and checks whether $w_i$ is near-a-black. Since this is indeed the case, a new black vertex is created in $\pi[w_i: w_{i+1})$, providing the sought contradiction.
Let $\overline{w}_i$ (resp. $\overline{w}_{i-1}$) be the vertex in $Q$ associated with $w_i$ (resp. $w_{i-1}$).
Notice that this argument also shows that, at time $t'_{i+1}$,  $\overline{w}_i$ is a child of $\overline{w}_{i-1}$ in $Q$ since $w_i$ becomes black after time $t'_i$ and not later than time $t'_{i+1}$, when $w_{i-1}$ was already black.

To prove (ii) it suffices to notice that, by inductive hypothesis, we only need to argue about the nodes in $\pi(w_i:w_{i+1})$.
From (*) we know that these nodes are white at time $t'_{i+1}$, while (i) ensures that $w_i$ is black at time $t'_{i+1}$. Then,  a similar argument to the one used in the base case shows that Algorithm~\ref{alg:add-leaf} will never color any node in $\pi(w_i:w_{i+1})$ black (as long as the nodes in $\pi$ remain uncorrupted). This concludes the proof by induction.

Let $w'$ be the node at distance $\delta$ from $\Tilde{u}$ in $\pi[\Tilde{u}:v]$.
Notice that $w_0$ belongs to $\pi[u:w']$. If $w_0$ lies in $\pi(\Tilde{u}:w']$, we can choose $w^* = w_0$. Otherwise,  $w_0$ is an ancestor of $\Tilde{u}$ and, from (i) and (ii), there is exactly one black vertex $b$ in $\pi(\Tilde{u}:w']$ and we choose $w^* = b$.
\end{proof}

\begin{algorithm}[t]
\caption{Answers a level ancestor query $\LA(v,k)$ in dynamic trees.}
\label{alg:dynamic-query}
\small
\If{$k \leq 7 \delta$}
{
	\Return $\texttt{climb}(v,k)$\;
}

\BlankLine

$\Tilde{v} \gets  \climb(v, \delta)$\;
Climb up the tree $\pseudoT$ from $\Tilde{v}$ for up to $ \delta$ nodes searching a black node\; 
\If{the previous procedure did not find a black node}
{
	\Return error;
}

\BlankLine

$v' \gets$ a black node found in the previous procedure\;
$d \gets$ distance between $v'$ and $v$\;
$k' \gets k -d-5\delta$\;
$u' \gets \LA_Q(q_{v'},\lfloor k'/\delta \rfloor)$\; 
$k_{\text{rest}} \gets k' -\lfloor k'/\delta\rfloor\cdot \delta+5\delta$\;
\Return $\texttt{climb}(u',k_{\text{rest}})$ \;
\end{algorithm}

The above lemma suggests a natural query algorithm. 
The query procedure is similar to the one for static case. When $k\le 7\delta$ we climb in $\pseudoT$  the nodes of the path from $v$ to the $k$-parent of $v$ in a trivial way. Otherwise, Lemma \ref{lem:blacknodes_pattern} ensures that if no vertex in the path $P$ from $v$ to its level ancestor in $T$ was corrupted by the adversary, then every other $\delta$-th vertex of $P$ is colored black except, possibly, for an initial subpath of length $\delta$ and for a trailing subpath of length $5 \delta$.
The query procedure explicitly ``climbs'' these portions of $P$ and queries $D_Q$ to quickly skip over its remaining ``middle'' part.
The pseudo-code is given in Algorithm~\ref{alg:dynamic-query}.

\noindent We are now ready to prove the main theorem of this section.

\begin{theorem}\label{thm:main_LA}
Our data structure requires linear space, supports the \AddLeaf operation in $O(\delta)$ worst-case time, and can answer resilient \LA queries in $O(\delta)$ worst-case time.
\end{theorem}
\begin{proof}
The correctness of the query immediately follows from Lemma \ref{lem:blacknodes_pattern}. Moreover, the time required to perform an \AddLeaf or an \LA operation is $O(\delta)$ since in both cases $O(\delta)$ vertices of $\pseudoT$ are visited and a single $O(\delta)$-time operation involving $D_Q$ is performed.

We now discuss the size of our data structure. Clearly, the space used to store the vector $\mathcal{A}$ of all records is $O(n)$. We only need to argue about the size of $D_Q$. Recall that $D_Q$ is the resilient version, obtained using the replication strategy, of the data structure that requires linear space, takes constant time to answer each $\LA$ query and to perform each $\AddLeaf$ operation \cite{AH00}. In order to show that $D_Q$ requires $O(n)$ space we will argue that the number of black vertices is $O(\frac{n}{\delta})$. As consequence we have that the size of $D_Q$ is $O(n)$.

To bound the number of vertices in $Q$, notice that in order to add a new vertex to $Q$ we need to spend $\delta$ flags that were previously unspent. Moreover,
a spent flag never becomes unspent unless the adversary corrupts the record of the corresponding node 
(by using one of its $\delta$ corruptions). As a consequence the nodes in $Q$ are at most $(n+\delta)/\delta=O(n/\delta)$.
\end{proof}

\section{Handling weighted \LA, \LCA, and \BVQ queries}\label{sec:handling_queries}

\subsection{Weighted \LA queries.}
In this section we show how to handle weighted $\LA$ queries when $\delta$ and the weights of the nodes are polylogarithmically-bounded positive integers.
Recall that, the answer to a weighted $\LA$ query $\LA(v, k)$ is the deepest ancestor $u$ of $v$ in $T$ such that the total weight of the vertices in the path from $u$ to $v$ in $T$ is \emph{at least} $k$.
The record of each node $v$ stores, along to the fields described in Section~\ref{sc:LA_dynamic}, an additional field containing the weight of $v$.
To store $Q$ we use a data structure $D_Q$ that maintains a forest of rooted trees in which every vertex has an associated weight. $D_Q$ is also able to answer weighted $\LA$  queries on $Q$ in $O(\delta)$ time.
For technical convenience we assume that a weighted level-ancestor query $\LA_Q(v, k)$ in $Q$ reports the vertex $u$ of minimum depth among the ancestors of $v$ such that the total weight of the vertices in the path between $u$ and $v$ (endpoints included) in $Q$ is \emph{at most} $k$.
This data structure is the resilient version of the one in \cite{AH00} which answers weighted $\LA$ queries in constant time when vertex weights are polylogarithmically-bounded. 

We modify Algorithm~\ref{alg:add-leaf} in two ways: (i) during the discovery phase, we keep track of the total weight $W$ of the vertices between $x$ (included) and the closest black proper ancestor $y'$ of $y$ ($y'$ excluded); (ii) during the execution phase, we keep track of the total weight $W'$ of the vertices between $x$ (included) and $x'$ (excluded).
Recall that when a vertex $x'$ becomes black in the execution phase of in Algorithm~\ref{alg:add-leaf} since $y$ was observed to be near-a-black in the discovery phase, the corresponding  vertex $\overline{x}'$ is added to $Q$ via the $\AddLeaf_Q$ operation on line~\ref{ln:add-leafQ-nearablack}. 
To handle weighted $\LA$ queries, we also need to assign a weight the the new vertex $\overline{x}'$.
Specifically, we choose this weight to be $W-W'$. Notice that, in the absence of corruptions, $W - W'$ is exactly the total weight of the vertices in path between $x'$ (included) and $y'$ (excluded).
Moreover, when a vertex $y$ becomes black black in the execution phase of in Algorithm~\ref{alg:add-leaf} because it was observed to be black-free in the discovery phase, we set the weight of the corresponding node $\overline{y}$ in $Q$ to the weight of $y$ in $T$.\footnote{Notice that, in absence of corruptions, the weight assigned to a vertex in $Q$ is always positive and polylogarithmically-bounded since so is $\delta$. The above property might not be true when the observed values are corrupted but we artificially enforce it by constraining  weights to be in this range. This also applies to weights read by the query algorithm explained in the following.}

We now describe how to answer a query $\LA(v,k)$. 
We start by optimistically assuming that the (unweighted) distance between $v$ and the sought vertex is short. 
We do so by climbing (up to) $10\delta$ levels from $v$ while keeping track of the total weight of the traversed vertices.
We stop at and return the first encountered vertex for which such a weight is at least $k$.

If the above procedure is unable to locate the sought vertex, we proceed as follows.
We climb $\delta$ levels from $v$, and we then search for a nearby black node $b$ among the closest $\delta$ proper ancestors of the reached vertex. 
During this process, we keep track of the total weight $W$ of the traversed nodes between $v$ (included) and $b$ (excluded).
Let $\overline{b}$ be the vertex in $Q$ that is associated with $b$.
We now perform an $\LA_Q(\overline{b}, k-W)$ query $D_Q$ to find the shallowest ancestor $\overline{b}'$ of $\overline{b}$ such that the overall weight $W'$ of the vertices in the path between $\overline{b}'$ and $\overline{b}$ in $Q$ is at most $k-W$. Let $b'$ be the node of $T$ that is associated to $\overline{b}'$.
Finally, we iteratively climb from $b'$ towards its ancestors until we reach a vertex $u$ such that the total weight of the path between $b'$ (excluded) to $u$ (included) is at least $k-W-W'$. We then return $u$. 
As we argue below, this final climbing procedure requires at most $6 \delta$ steps in the absence of corruptions. Therefore, if this threshold is exceeded we immediately stop the query and report an error.

We now discuss the correctness of the query. Let $u^*$ be the deepest ancestor of $v$ in $T$ such that the total weight between $v$ and $u^*$ is at least $k$, and assume that the path $\pi$ between $u^*$ and $v$ is uncorrupted.
To prove that the query procedure is correct it is sufficient to show that (i) the vertex $b'$ belongs $\pi$, and (ii) the (unweighted) length of $\pi[u^*:b']$ is at most $6 \delta$. 

To see (i), assume by contradiction that $b'$ in not in $\pi$, and let $\overline{b}_1$ the deepest ancestor of $\overline{b}$ in $Q$ such that the vertex in $T$ corresponding to the parent $\overline{b}_2$ of $\overline{b}_1$ in $Q$ does not belong to $\pi$. Let $b_2$ the vertex in $T$ associated to $\overline{b}_2$ (vertex $\overline{b}_2$ must exists since we assumed that $b'$ is not in $\pi$). As a consequence, since $\pi$ is uncorrupted, the total weight of the path in $Q$ between $\overline{b}_1$ (excluded) and $\overline{b}$ is equal to the total weight of $\pi(b_1,b]$. Moreover, 
the weight of $\overline{b}_1$ in $Q$ is at least the total weight of $\pi[u^*:b_1]$. This implies that $W'$ must be strictly greater than $k-W$ since $\overline{b}_2$ has weight at least $1$. This is a contradiction. 

It remains to prove (ii). Since the number of vertices of $\pi$ is at least $10 \delta$, we invoke Lemma~\ref{lem:blacknodes_pattern} to conclude that $b'$ must be at distance at most $6 \delta$ from $u^*$. Indeed, if this was not the case, then the $\delta$-parent of $b'$ would be black and would belong to $\pi$, which implies that $\overline{b}'$ could not be the vertex returned by the query on $D_Q$.

\subsection{\BVQ queries.}

The record of each node $v$ stores the weight of $v$ and maintains an additional field $\depth_v$ that intuitively keeps track of the depth of $v$ in $T$.
Initially, when $T$ is first built and consists only of the root $r$, we set $\depth_r = 0$.
Whenever a new node $v$ is appended as a child of $u$ via the $\AddLeaf$ operation, we set $\depth_v = \depth_u + 1$.

To store $Q$ we use a data structure $D_Q$ that that maintains a forest of rooted trees which can be updated by adding leaves in $O(\delta)$ time per operation. $D_Q$ is also able to answer (unweighted) $ \LA$ and $\BVQ$ queries on $Q$ in $O(\delta)$ time. This data structure is the resilient version of the one in \cite{AH00} which answers both $\LA$ and $\BVQ$ queries in constant time.

Moreover, we slightly modify the execution phase of Algorithm~\ref{alg:add-leaf} in the case in which $y$ was observed to be near-a-black in the discover phase.
In this scenario a vertex $x'$ becomes black, and a corresponding vertex $\overline{x}'$ is added to $Q$ via the $\AddLeaf_Q$ operation on line~\ref{ln:add-leafQ-nearablack}.
In our modification, we additionally climb the path from $x'$ (included) to $y'$ (excluded) while keeping track of the  vertex $w'$ of minimum weight among the encountered nodes.
We assign the weight of $w'$ to $\overline{x}'$ in $Q$ and we store a reference to $w'$ as (replicated) satellite data attached to $\overline{x}'$.
When instead the vertex that becomes black is $y$ since $y$ was observed to be black free during the discovery phase, we assign weight $+\infty$ to the corresponding black node $\overline{y}$ in $Q$ (no satellite data is needed in this case).

We now describe how to answer a $\BVQ(u,v)$ query.
In particular, we only need to consider the case in which $u$ is an ancestor of and $v$ since we can always perform a $\LCA(u,v)$ query (we will show how to handle $\LCA$ queries later in this section) to find the lowest common ancestor $z$ of $u$ and $v$ in $T$ and then return the minimum of the two bottleneck queries $\BVQ(z,u)$ and $\BVQ(z,v)$ which satisfy the above requirement.

Hence we assume that no corrupted vertex exists in the path $\pi$ from $u$ to $v$ in $T$ and we start by computing the quantity $d = \depth_u - \depth_v$.
Notice that, while $\depth_u$ (and $\depth_v$) might not contain the actual depth of $u$ (and $v$) in $T$, due to a corruption in some ancestor of $u$, the value of $d$ always matches the distance between $u$ and $v$ in $T$, i.e., the length of $\pi$.\footnote{The adversary could cause the values stored in $\depth_u$ or $\depth_v$ to overflow. However, with some additional technical care (by interpreting these fields  as unsigned integers in modular arithmetic) one can always recover $d$.}

If $d < 10\delta$, we answer the query using the trivial algorithm $\BVQ(u,v)$ that climbs the path $\pi$ one edge at a time from $v$ to $u$ and return the vertex of minimum weight encountered in the process. Clearly this algorithm is resilient and requires $O(\delta)$ time.
Otherwise, we use a strategy similar to the one used for $\LA$ queries in Algorithm~\ref{alg:dynamic-query}.
We climb $\delta$ levels from $v$, and we then search for a nearby black node $b$ among the $\delta$ ancestors of the reached vertex.
During this process, we keep track of the node $w_1$ of minimum weight among those we encounter.
Next, we perform an $\LA$ query on $D_Q$ to find the black node $b'$ that is $(\lfloor d / \delta \rfloor - 7)$-th parent of $b$ (Lemma~\ref{lem:blacknodes_pattern} ensures that this vertex exists and is black). Finally, we climb from $b'$ to $u$ in $O(\delta)$ time and keep track of a node $w_2$ having minimum weight among those encountered during the process.
The answer to $\BVQ(u,v)$ is the vertex of minimum weight between $w_1$, $w_2$, and the node returned by a bottleneck vertex query $\BVQ(b',b)$ in $Q$.

\subsection{\LCA queries.}
\newcommand{\cba}{\texttt{cba}\xspace}

The record of each node $v$ stores, along with the fields described in Section~\ref{sc:LA_dynamic}, a field $\depth_v$ managed as discussed for the $\BVQ$ query, and an additional field $\cba_v$ which intuitively stores a pointer to the vertex in $Q$ associated with the closest black ancestor of $v$ in $T$. 
When $v$ is inserted $\cba_v$ is \emph{unset}, and it will be possibly set during the execution phase of some later $\AddLeaf$ operation.
Similarly to $\nodeflag_v$, we allow a field $\cba_v$ that is unset to be annotated with a pair $(x,i)$, where $x$ is a vertex that is being inserted and $i$ is the observed distance between $x$ and $v$. 

To store $Q$ we use a data structure $D_Q$ that maintains a forest of rooted trees which can be updated by adding leaves in $O(\delta)$ time per operation. $D_Q$ is also able to answer $\LA$ and $\LCA$ queries on $Q$ in $O(\delta)$ time. This data structure can be obtained as combination of the resilient versions of the ones in \cite{AH00,CH05} which answers $\LA$ and $\LCA$ queries in constant time.

We modify both the discovery and the execution phases of Algorithm~\ref{alg:add-leaf}. Recall that in the discovery phase the algorithm locates the $\delta$-parent $y$ of $x$, and the closest proper black ancestor $y'$ of $y$, if any.
In our modification, when we traverse a generic ancestor $z$ of the inserted vertex $x$, we check $\cba_z$. If $\cba_z$ is unset, we annotate it with $(x, i)$ where $i$ is the observed distance between $x$ and $z$ (possibly overwriting previous annotation).
Moreover, we also store $q_{y'}$ in a variable $\overline{y}'$ in safe memory. 
In the execution phase, we only need to handle the case in which $y$ appeared to be near-a-black during the discovery phase. In this case, let $x'$ be the vertex such that $y'$ is the $\delta$-parent of $x'$ (see line~\ref{ln:execution-reachingx'}). 
In this case, we extend the for loop of line~\ref{ln:phase2-flag1} in order to reach $y'$. We still check and spend the encountered flags only for the fist $\delta$ vertices as before. In addition, for each vertex $z$ at distance $i$ from $x$, such that $z$ in the path between $x'$ (included) and $y'$ (excluded), we check that $\cba_z$ is either set to $\overline{y}'$ or unset and (correctly) annotated with $(x,i)$. In the latter case, we set $\cba_z$ to $\overline{y}'$. 
If neither of the previous conditions is met (i.e., $\cba_z$ is set to some vertex other than $\overline{y}'$ or it is unset and incorrectly annotated) we are in an \emph{exceptional situation} and $\cba_z$ is left unaltered.

Finally, we modify line~\ref{ln:add-leafQ-nearablack} in which $x'$ is colored black via the addition of a corresponding vertex $\overline{x}'$ to $Q$.
Our modification is as follows: If we are not in an exceptional situation, we proceed as before and we add $\overline{x}'$ as child of $\overline{y}'$ in $Q$. Otherwise, in the exceptional situation, we add $\overline{x}'$ as a new root in $Q$.

Before describing how to answer to a $\LCA$ query, we argue
that the above modifications guarantee stronger structural properties than the ones given in Section~\ref{sc:LA_dynamic}. In particular, we show that Lemma~\ref{lem:blacknodes_pattern} still holds.
First of all, notice that our modifications do not affect vertex colors. Hence, we only need to show that the parent-child relationships between black vertices in $Q$ are preserved.
Since the only way to alter these relationships is for an exceptional situation to happen during the execution of $\AddLeaf$ that colors some node $x'$ black, 
we only need to show that no exceptional situation can arise
when a (sufficiently deep) vertex of an uncorrupted path becomes black. This is proven in the following.

\begin{lemma}
Let $\pi$ be an ancestor-descendant path in $T$ of length at least $2\delta$, and let $x'$ be the deepest vertex of $\pi$.  
If $x'$ is black and no vertex in $\pi$ has been corrupted, then the execution of $\AddLeaf$ that colored $x'$ black did not encounter an exceptional situation.
\end{lemma}
\begin{proof} 
Let $x$ be the node whose insertion in $T$ cause $x'$ to be colored black, and let $t_x$ the time immediately before $x$ is inserted in $T$. In the rest of the proof, we assume that the execution of $\AddLeaf$ inserting $x$ in $T$ is in an exceptional situation, and we prove that this leads to a contradiction. 

Since we are in an exceptional situation, at time $t_x$ the $\delta$-parent $y'$ of $x'$ must be black and all the other nodes in $\pi(y':x')$ must be white.  Let $\overline{y}' = q_{y'}$.
Then, the exceptional situation was caused by a node $w$ in $\pi(y':x']$ such that $\cba_w=\overline{z}$ and $\overline{z} \neq \overline{y}'$.
Let $z$ be the node in $T$ that is associated with vertex $\overline{z}$ and notice that  $\overline{z} \neq \overline{y}'$ implies $z \neq y'$.
Since no vertex in $\pi$ is corrupted, the existence of $w$ implies the existence of an ancestor $z$ of $w$ which is black at time $t_x$ and such that $d(z,w) \le \delta$. By hypothesis,  all the nodes in $\pi(y':x')$ are white at time $t_x$, and hence $z \neq y'$ must be a proper ancestor of $y'$. Node $z$ satisfies the following conditions: (i) $d(y',z) \leq \delta-1$ (since $d(w,z) \leq \delta$), and (ii) $y'$ was white when $\cba_w$ was set to $\overline{z}$ (since  $\cba_w=\overline{z}$ and no vertex in $\pi$ is corrupted).
This implies that, when $y'$ was colored black, there was a black node $z$ such that $d(y',z) \leq \delta-1$ and this is a contradiction.
\end{proof}

We now prove a structural property that will be exploited in the query procedure. 
If we need to answer a $\LCA(u,v)$ query and both $u$ and $v$ are black, we can relate the lowest common ancestor of $u$ and $v$ in $T$ with the  lowest common ancestor of the corresponding vertices  $\overline{u}$ and $\overline{v}$ (respectively) in $Q$.
More precisely, let us assume that the path $\pi$ between $u$ and $v$ in $T$ is uncorrupted, and let $w$ be the lowest common ancestor of $u$ and $v$. 
Let $\overline{a}$ (resp. $\overline{b}$) be the the shallowest ancestor of $u$ (resp $\overline{v}$) in $Q$ such that the corresponding vertex $a$ (resp. $b$) in $T$ belongs to $\pi[w:u]$ (resp. $\pi[w:v]$).

We prove the following lemma.
\begin{lemma}\label{lm:a-b-close}
The distance between $w$ and $a$ (resp. $b$) in $T$ is at most $6\delta$. 
Moreover, if both $\overline{a}$ and $\overline{b}$ have a parent in $Q$, then such parents coincide.  
\end{lemma}
\begin{proof}
The bound on the distance between $w$ and $a$ (resp. $b$) in $T$ immediately follows from Lemma~\ref{lem:blacknodes_pattern}, hence in the rest of the proof we focus on showing that that the parents of $\overline{a}$ and $\overline{b}$ (if they exist) must coincide.

Assume, w.l.o.g., that $\overline{a}$ was inserted in $Q$ before $\overline{b}$, and let $\overline{y}'$ be the parent of $\overline{a}$ in $Q$. 
Notice that, when $a$ was colored black, the corresponding $\AddLeaf$ operation inserted vertex $\overline{a}$ as a child of $\overline{y}'$.
Hence, the black vertex $y'$ in $T$ corresponding to $\overline{y}'$ was observed to be an ancestor of both $a$ and $w$ (by the choice of $\overline{a}$) in the discovery phase. As a consequence, the execution phase ensured that the value of $\cba_w$ is exactly $\overline{y}'$ (as otherwise the $\AddLeaf$ operation would have encountered an exceptional situation and $\overline{a}$ would have been a root of a tree in $Q$). 
Analogously, the $\AddLeaf$ operation that colored $b$ black was not in an exceptional situation and, by the definition of $\overline{b}$, we know that $w$ lies in the (observed) path between $b$ and its (observed) $\delta$-parent $z$.
Since $w$ is uncorrupted, $\AddLeaf$ operation successfully checked that $\cba_w$ matched $q_{z}$, we conclude that $q_z=\overline{y}'$. Hence, the parent of $\overline{b} = \overline{y}'$.
\end{proof}

We are now ready to describe how to answer to a $\LCA(u,v)$ query. 
We first describe a simple \emph{naive query} that correctly handles the case in which at least one of $u$ and $v$ is close to their lowest common ancestor $w$.
If this is not the case, this query will be \emph{inconclusive}.

Let $k$ be the difference between the depth of $u$ and the depth of $v$.\footnote{Similarly to the case of $\BVQ$ queries, this value can be recovered from $\depth_v$ and $\depth_u$.}
We describe the case $k\ge 0$ (the case $k<0$ is symmetric).
We perform an $\LA(v, k)$ query to find the $k$-parent $v'$ of $v$. Notice that, in absence of corruptions, the distance between $w$ and $u$ is the same as the distance between $w$ and $v'$.
We now iteratively perform the following steps.
We check whether $v'=u$, if this is the case we answer the query by reporting $v'$ as the sought lowest common ancestor. Otherwise, we move $u$ and $v'$ to their respective parents and repeat.
If the parent of $u$ or $v$ does not exist or we are unable to answer the query within $10\delta$ iterations, we stop the above procedure and say that the naive query is inconclusive. 

We now need to handle the case in which the naive query is inconclusive.
For simplicity we first describe our strategy assuming that both $u$ and $v$ are black, and then we show how to handle a generic query. 

Let $\overline{u}$ and $\overline{v}$ be the vertices in $Q$ corresponding to $u$ and $v$, respectively. We perform an $\LCA$ query in $Q$ to find the lowest common ancestor $\overline{y}'$ of $\overline{u}$ and $\overline{v}$, if any. We assume also that, if such a vertex exists, this query is able to return the two vertices $\overline{a}$ and $\overline{b}$ of $Q$ such that $\overline{y}'$ is the parent of both $\overline{u}$ and $\overline{v}$.\footnote{It is easy to support such a query with a constant number of $\LCA$ and $\LA$ queries on $Q$.} If the $\overline{y}'$ exists, we return the outcome of the naive $\LCA$ query on $a$ and $b$, where $a$ (resp. $b$) is the black vertex in $T$ corresponding to $\overline{a}$ (resp. $\overline{b}$).
Lemma~\ref{lm:a-b-close} ensures that the vertices $a$ and $b$ are close descendants of $w$, and hence the naive query correctly finds $w$.

It remains to handle the case in which $\overline{y}'$ does not exist, i.e., $\overline{u}$ and $\overline{v}$ belong to different trees of $Q$.
In this case, we let $\overline{a}'$ (resp. $\overline{b}'$) be root of the tree in $Q$ that contains $\overline{u}$ (resp. $\overline{v}$). From Lemma~\ref{lm:a-b-close}, it must be that $\overline{a}'=\overline{a}$ or $\overline{b}'=\overline{b}$ (possibly both). 
In this case, we inspect the fields $\depth_{u}$, $\depth_{v}$, $\depth_{a'}$ and $\depth_{b'}$ and we consider the vertex among $a'$ and $b'$ that appears to be deeper. W.l.o.g., let $a'$ be such vertex and let $k_{a'}$ be the observed difference in levels between $u$ and $a'$. 
We check that $k_{a'}$ is non-negative and that $\LA(u, k_{a'}) = a'$. If the above condition is met, we perform a naive $\LCA$ query between $a'$ and $v$. If this query answers with some vertex $w'$, it must be that $w'=w$ and hence we return it. 
Otherwise, if $k_{a'} <0$, the answer $\LA(u, k_{a'})$ was not $a'$, or the naive query was inconclusive, we must have $b' = b$ and we return the result of the naive $\LCA$ query between $b'$ and $v$.

It remains to handle the case in which at least one among $u$ and $v$ is white. Since we already performed a naive query between $u$ and $v$, we know that both distances between $u$ and $w$, and $v$ and $w$ are more than $10 \delta$. Hence, we climb from $u$ (resp. $v$) until we reach the first black ancestor $u'$ (resp. $v'$) of $u$ (resp. $v$).
By Lemma~\ref{lem:blacknodes_pattern}, we must find such a node $u'$ (resp. $v'$)  at distance at most $2\delta$ from $u$ (resp. $v$).
We can now return the vertex reported by a $\LCA(u', v')$ query, which can be answered as described above since both $u'$ and $v'$ are black.

\clearpage

\bibliographystyle{plainurl}
\bibliography{njl}
\end{document}